\documentclass[11pt]{article}

\usepackage[utf8]{inputenc}
\usepackage[T1]{fontenc}
\usepackage{amsmath,amssymb,amsthm}
\usepackage{mathtools}
\usepackage{xcolor}
\usepackage{tikz}
\usepackage{graphicx}
\usepackage{overpic}
\definecolor{C0}{HTML}{1f77b4}
\definecolor{C1}{HTML}{ff7f0e}
\definecolor{C2}{HTML}{2ca02c}
\definecolor{C3}{HTML}{d62728}
\definecolor{C4}{HTML}{9467bd}
\definecolor{C5}{HTML}{8c564b}
\definecolor{C6}{HTML}{e377c2}
\definecolor{C7}{HTML}{7f7f7f}
\definecolor{C8}{HTML}{bcbd22}
\definecolor{C9}{HTML}{17becf}
\usepackage[colorlinks=true, linkcolor=blue!60!black, citecolor=blue!60!black, urlcolor=blue!60!black]{hyperref}
\usepackage{geometry}
\newtheorem{theorem}{Theorem}[section]
\newtheorem{proposition}[theorem]{Proposition}
\newtheorem{lemma}[theorem]{Lemma}

\theoremstyle{definition}
\newtheorem{definition}[theorem]{Definition}

\newtheorem{remark}[theorem]{Remark}

\newcommand{\cB}{\mathcal{B}}
\newcommand{\cC}{\mathcal{C}}
\newcommand{\cD}{\mathcal{D}}
\newcommand{\cT}{\mathcal{T}}
\newcommand{\cZ}{\mathcal{Z}}
\newcommand{\bC}{\mathbb{C}}
\newcommand{\bZ}{\mathbb{Z}}
\newcommand{\bD}{\mathbb{D}}
\newcommand{\Bord}{\mathrm{Bord}}
\newcommand{\Hilb}{\mathrm{Hilb}}
\newcommand{\Vect}{\mathrm{Vect}}
\newcommand{\op}{\mathrm{op}}
\newcommand{\SO}{SO}
\newcommand{\Or}{O} 
\newcommand{\dg}{\dagger}

\title{Topological defects in reflection positive topological field theories}
\author{Lukas Müller, LMU Munich}
\date{}

\begin{document}
\maketitle

\begin{abstract} 
Topological defects in quantum field theories are believed to assemble into higher categories with extra structure. This has been made precise for defects in $2$- and $3$-dimensional oriented topological quantum field theories by Davydov, Kong, and Runkel and by Carqueville, Meusburger, and Schaumann, respectively. In this paper we study the extra structure present on these categories when the topological field theory is additionally reflection positive. We define reflection defect TQFTs as symmetric monoidal functors out of a defect bordism category that intertwine orientation reversal with complex conjugation; they are \emph{reflection positive} if they satisfy an additional positivity condition. Our main result is that in two dimensions the bicategory of defects $\cT_\cZ$ associated to a reflection defect TQFT carries the natural structure of an $\Or(2)$-dagger bicategory, a structure we define explicitly. If the theory is reflection positive, $\cT_\cZ$ can be equipped with additional structure closely related to the definition of a 3-Hilbert space (the two agree up to some finiteness and completeness conditions).
\end{abstract}

\tableofcontents

\section{Introduction}
\label{sec:Intro}
Topological defects in quantum field theory have been a subject of intense study over the last decades, with renewed recent interest because of their connection to generalized symmetries~\cite{BBFT24,CDIS22,FMT24,GKSW15,SN23,Sha23}. Such defects in an $n$-dimensional topological field theory assemble into an $n$-category $\cD$ with additional structure reflecting the tangential structures the defects carry. From both a physical and a mathematical perspective it is desirable to include unitarity in the picture. The required theory of unitary higher categories is in its infancy, and there are many open questions in this context. On the physics side, the action of generalized symmetries on extended operators---their generalized charges---is captured by higher representation theory and via the Symmetry TFT~\cite{BBG23a,BBG23b,BSN23a,BSN23b,CHFHS24}; its unitarity has recently been studied by Bartsch~\cite{Bar25} through $\ast$-representations of the associated tube algebras. On the mathematical side, the theory of higher dagger categories has recently emerged as a language to formulate these structures, as well as the general structure of topological defects~\cite{FHJF24, CL25, Mue25}. For oriented defects in oriented 2-dimensional theories this additional structure is a pivotal structure on the defect bicategory~\cite{DKR11,Car16}, which can be identified with an $\SO(2)$-dagger structure in the sense of~\cite{CL25, Mue25}; in three dimensions defects form a Gray category with duals~\cite{CMS}.

Topological quantum field theories were formulated by Atiyah~\cite{Ati88} as symmetric monoidal functors $\cZ \colon \Bord_n \to \Vect_\bC$, where $\Bord_n$ is the symmetric monoidal category whose objects are closed oriented $(n-1)$-dimensional manifolds and whose morphisms from $\Sigma$ to $\Sigma'$ are (equivalence classes of) $n$-dimensional oriented compact manifolds with boundary identified with $\overline{\Sigma} \sqcup \Sigma'$, composed by gluing, where $\overline{\Sigma}$ denotes the orientation reversal of $\Sigma$. To capture locality fully one allows the objects of $\Bord_n$ to be cut further, leading to the notion of a fully extended topological quantum field theory, which assigns higher categorical data to manifolds of higher codimension~\cite{BD,Lur}. Understanding unitarity in this extended setting is an active research direction.

In this short paper we approach the problem from a slightly different direction, providing a new avenue for insight and one part of a fuller picture that has yet to emerge in full generality. It rests on the operator--state correspondence: the higher categories formed by defects give access to the categorical data an extended topological quantum field theory assigns in higher codimension~\cite{Kap10}, so studying defects in non-extended topological field theories offers another route towards the same data.

The question we answer in this paper is: what is the higher categorical structure corresponding to topological defects in non-extended unitary (or reflection positive) topological quantum field theories? We focus on oriented 2-dimensional theories---treating the elementary 1-dimensional case first as a warm-up---and show that their topological defects form an $\Or(2)$-dagger bicategory satisfying additional positivity conditions, essentially those appearing in the definition of a 3-Hilbert space in~\cite{CHFHS24}.      

To define reflection positive topological field theories with defects, we combine two strands of the literature. The first is the defect bordism categories $\Bord_n^{\textrm{def}}(\bD)$ of~\cite{CRS18}, whose objects are stratified $(n-1)$-dimensional manifolds labelled by defect data $\bD$ and whose morphisms are stratified labelled cobordisms between them; an example of such a morphism in two dimensions is:
\begin{equation} 
\label{eq:DefectBordIntro}
\vcenter{\hbox{
		\begin{overpic}[scale=0.5]{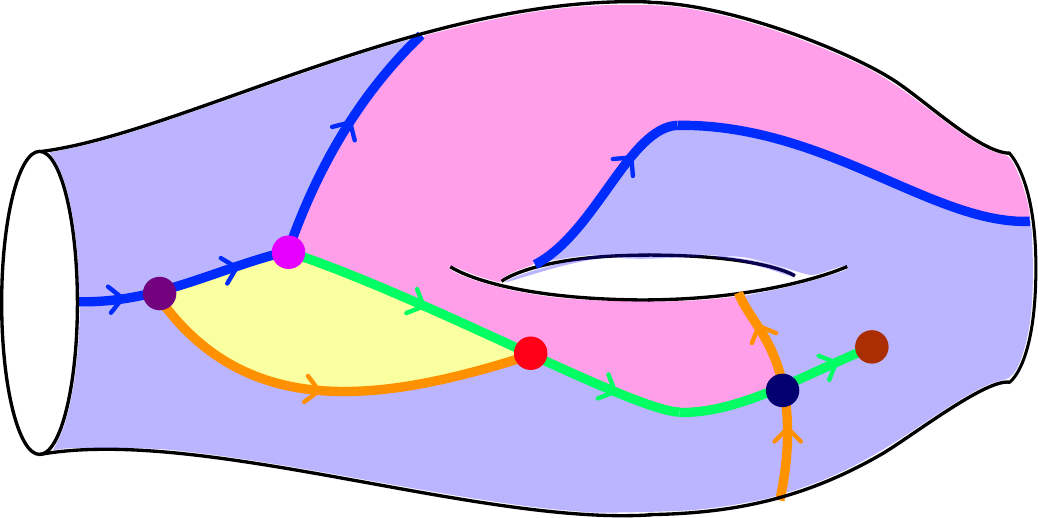}
			\put(50,5){$u_1$}
			\put(92,22){$u_1$}
			\put(15,34){$u_1$} 
			\put(32,15){$u_2$} 
			\put(45,36){$u_3$}
			\put(8,22.5){$X_1$}
			\put(30,8){$X_2$} 
			\put(20,26){$X_3$}
			\put(31.5,30){$X_{10}$}
			\put(38,23){$X_4$} 
			\put(63,12){$X_5$}
			\put(76,16.5){$X_6$}
			\put(75,38){$X_7$}
			\put(77,5){$X_8$}
			\put(66,17){$X_9$}
			\put(12,18){$\varphi_1$}
			\put(27,21.5){$\varphi_2$}
			\put(51.5,18){$\varphi_3$}
			\put(70.5,8.5){$\varphi_4$}
			\put(84,19){$\varphi_5$} 
			\put(13.5,23.5){\small $\color{C1}+$}
			\put(29.5,25.5){\small$\color{C2}-$}
			\put(49.9,11.5){\small$\color{C3}+$} 
			\put(77,10){\small$\color{C4}+$}
			\put(86,14.5){\small$\color{C5}-$}
\end{overpic}}}
\end{equation} 
Here the shaded $2$-strata carry bulk labels $u_i$, the coloured $1$-strata are defect lines labelled by $X_i$ and carrying orientations and coorientations, and the points $\varphi_i$ where they meet are the junctions.

The second is the description of non-extended reflection positive topological field theories in the formulation of Freed--Hopkins, and its reformulation in terms of dagger categories~\cite{SS23,Ste23}.
Reflection positive topological field theories (without defects) were defined in the influential work of Freed and Hopkins~\cite{FH16} as the combination of two distinct parts: in a \emph{reflection theory} the geometric operation of orientation reversal (or an appropriate generalization for other tangential structures) is intertwined with the algebraic operation of complex conjugation, while \emph{positivity} is the further property that certain hermitian pairings, defined using the reflection structure, are positive definite.   

Incorporating topological defects into this definition of reflection positivity is straightforward (though to the best of our knowledge this is the first place where this has been done): reversing the orientation of a stratified manifold together with all of its strata gives a $\bZ_2$ reflection action on $\Bord_n^{\textrm{def}}(\bD)$. A \emph{reflection defect topological field theory} is then simply a $\bZ_2$-equivariant symmetric monoidal functor $\Bord_n^{\textrm{def}}(\bD) \to \Vect_\bC$ (where $\bZ_2$ acts on $\Vect_\bC$ by complex conjugation), and reflection positivity is again the positivity of certain pairings coming from the reflection structure. We set up these reflection positive defect bordism categories in arbitrary dimension $n$ (Section~\ref{sec:proposal}) before specializing to $n = 1, 2$.  

In two dimensions one can extract from this data a bicategory of topological defects: its objects are the labels for 2-strata, its 1-morphisms are lists of labels for 1-dimensional strata, and its 2-morphism spaces are the state spaces $\cZ(E_{X,Y})$ the theory assigns to certain decorated circles $E_{X,Y}$, encoding the operator--state correspondence; composition is supplied by evaluating the defect theory on a pair of pants. We denote this defect bicategory by $\cT_\cZ$. A reflection structure provides coherent isomorphisms $\cZ(E_{X,Y}) \cong \overline{\cZ(E_{Y,X})}$, which we show assemble into a functor $\dg \colon \cT_\cZ \to \cT_\cZ^{2\op}$ that is the identity on objects and 1-morphisms (the superscript $2\op$ meaning $\dg$ reverses the direction of 2-morphisms). All 1-morphisms have both left and right adjoints, which agree and are given by reversing the orientations of the defect lines; these adjoints are compatible with $\dg$, forming a unitary adjoint functor in the sense of~\cite{Pen20,CHFHS24}.

We call the resulting structure an $\Or(2)$-dagger bicategory, a name suggesting how it should fit into the larger framework of higher dagger categories labelled by subgroups of $\Or(2)$~\cite{FHJF24,Mue25}. We define this structure in detail in Section~\ref{sec:strictO2}; up to nomenclature, most of it is already contained in~\cite{CHFHS24}.  

The main result of Section~\ref{sec:extractionsection} is the following. 
\begin{theorem}
The bicategory of topological defects in a 2-dimensional reflection defect topological field theory forms an $\Or(2)$-dagger bicategory. 
\end{theorem}   
 
 The two generators of this $\Or(2)$-structure are the pivotal structure (the $\SO(2)$ part) and the dagger (the reflection), and we \emph{derive} both from the geometry of 2-dimensional defect bordisms and their reflections; their compatibility expresses the fact that a reflection conjugates a rotation to its inverse.

In Section~\ref{sec:positivity} we move on to study the additional structure this $\Or(2)$-dagger bicategory acquires when the defect theory is reflection positive; we essentially find that of a pre-3-Hilbert space, without some of the finiteness assumptions. If one adds these assumptions, topological defects in reflection positive TQFTs form a 3-Hilbert space.
 
 Our results should carry over to topological defects in arbitrary quantum field theories along the lines of~\cite{DKR11}. In the case of a single defect label for 2-dimensional strata, the structure we find is that of a unitary fusion category (up to completeness and finiteness conditions) equipped with a unitary dual functor, recovering the expected categorical structure for unitary categorical symmetries of 2-dimensional quantum field theories.

\paragraph{Statement on AI use.} (The following paragraph was written by Claude without additional modifications or changes by the author)
This paper was written in close collaboration with Claude (Anthropic), primarily the models Claude Opus and Claude Fable, and was in part an experiment in whether current AI systems are useful collaborators in mathematical research. The work was carried out in a series of sessions between June and July 2026, with a persistent project memory carrying the mathematical state of the draft between sessions. The mathematical direction, the key structural decisions, and the central ideas are due to the human author, who also edited the text directly between sessions. Claude contributed drafts of most sections in the author's style, worked out and checked proof details, proposed simplifications and caught gaps (for instance a missing coherence axiom in an early version of the definition of an $\Or(2)$-dagger bicategory), produced all TikZ figures---several of them from hand drawings---performed literature searches with verification of all added references, and maintained the \LaTeX{} source. The collaboration was iterative in both directions: drafts by Claude were corrected by the author, and text by the author was polished by Claude. All definitions, statements, and proofs have been verified by the human author, who takes full responsibility for the content.

\paragraph{Comments by the author.}
Here I record some comments on the process for anyone who might consider doing something similar.
In my opinion, the experiment was successful, for the following reasons. The goal of the project was to combine two well-established approaches in the literature, which are well documented and explained. After understanding these, the project itself isn't that hard. I did have a good understanding of all the mathematics involved prior to the project, which made it easy to spot mistakes and structure the project. In many regards I approached the writing similarly to how one might approach the supervision of a student project. We started by reading and summarizing the necessary background papers. The two essential parts are the Freed--Hopkins definition of reflection positive TQFTs and the definition of defect TQFTs by Carqueville, Runkel, and Schaumann, as well as the procedure to extract defect bicategories in 2 dimensions. Afterwards we drafted a definition of reflection positive defect TQFTs (this required a few iterations and fixing of mistakes). Then we worked towards proving the main theorem, which required many iterations, simplifications of arguments, and fixing of mistakes. Certain things just didn't work, like drawing Figure~\ref{fig:spherical}. After multiple tries I ended up drawing the figure myself.

For now, I will not adopt this approach for other, more involved projects and would warn anyone against using AI in this way for research in areas they aren't an expert in. It can be useful in many other ways, though. I do take full responsibility for mistakes and errors in this paper, of which probably many remain.

\paragraph{Acknowledgments.} 
I am grateful to Nils Carqueville and Luuk Stehouwer for helpful comments on an early draft of the article.

\section{\texorpdfstring{$\Or(2)$}{O(2)}-dagger bicategories}
\label{sec:strictO2}

In this section we give a concrete definition of ($\bC$-anti-linear) $\Or(2)$-dagger bicategories as suggested in~\cite[Example~3.6]{Mue25}\footnote{There the $\bC$-linear version was discussed, because it plays an important role in the description of unoriented topological defects in unoriented quantum field theories. The anti-linear version is what is relevant for unitary defect topological quantum field theories.} and to be studied in a coherent formulation in~\cite{MS25}. The structure consists of a $\bC$-linear bicategory with adjoints equipped with an anti-linear involution reversing the vertical direction of 2-morphisms and a unitary adjoint functor. This is exactly the many-object version of the unitary dual functors of~\cite{Pen20}, as considered in~\cite{CHFHS24}.

Throughout, $\cB$ is a $\bC$-linear bicategory with adjoints; \emph{$\bC$-linear} means that the 2-morphism sets are complex vector spaces and horizontal and vertical composition are bilinear, i.e.\ $\cB$ is enriched in $\bC$-linear categories.

\begin{definition}
\label{def:strictO2}
An \emph{$\Or(2)$-dagger structure} on $\cB$ consists of:
\begin{enumerate}
  \item a functor $\dg \colon \cB \to \cB^{2\op}$ which is the identity on objects and 1-morphisms, anti-linear on 2-morphism spaces, and satisfies $\dg^2 = \mathrm{id}_\cB$;
  \item a choice of \emph{dual functor} $(-)^L \colon \cB \to \cB^{1\op,2\op}$ which is the identity on objects: an assignment of a left adjoint $X^L \colon \beta \to \alpha$ with adjunction data $\mathrm{ev}_X \colon X^L \otimes X \Rightarrow 1_\alpha$ and $\mathrm{coev}_X \colon 1_\beta \Rightarrow X \otimes X^L$ to every 1-morphism $X \colon \alpha \to \beta$, acting on 2-morphisms by left mates,
  \begin{equation}
    \varphi^L := (\mathrm{ev}_Y \otimes \mathrm{id}_{X^L}) \circ (\mathrm{id}_{Y^L} \otimes \varphi \otimes \mathrm{id}_{X^L}) \circ (\mathrm{id}_{Y^L} \otimes \mathrm{coev}_X) \colon Y^L \Rightarrow X^L \, ,
  \end{equation}
  regarded as a 2-functor whose compositor $\nu_{X,\widetilde{X}} \colon X^L \otimes \widetilde{X}^L \Rightarrow (\widetilde{X} \otimes X)^L$ and unitor $\iota_\alpha \colon 1_\alpha \Rightarrow 1_\alpha^L$ are the canonical comparison 2-morphisms built from the adjunction data;
\end{enumerate}
such that $(-)^L$ is a \emph{dagger} 2-functor: it commutes with $\dg$, i.e.\ $(\varphi^\dg)^L = (\varphi^L)^\dg$ for all 2-morphisms $\varphi$, and the coherence 2-isomorphisms $\nu$ and $\iota$ are unitary ($\nu^\dagger=\nu^{-1}$ and $\iota^\dagger=\iota^{-1}$).
\end{definition}

Unpacking datum~(1): $\dg$ reverses vertical composition, $(\Psi \circ \Phi)^\dg = \Phi^\dg \circ \Psi^\dg$, and preserves horizontal composition, $(\widetilde{\Phi} \otimes \Phi)^\dg = \widetilde{\Phi}^\dg \otimes \Phi^\dg$. In particular each hom-category $\cB(\alpha, \beta)$ is an anti-linear dagger category---i.e.\ a $\bC$-linear category equipped with an anti-linear dagger $\dg$. The final condition adds, on top of the adjunctions of datum~(2), that the mate construction is compatible with $\dg$ and that the resulting comparison 2-morphisms are unitary.

\begin{remark}[Strictness]
\label{rem:strictcase}
An $\Or(2)$-dagger structure is \emph{strict} if $1_\alpha^L = 1_\alpha$ and $(\widetilde{X} \otimes X)^L = X^L \otimes \widetilde{X}^L$, and the adjunction data is \emph{normalized and multiplicative},
\begin{equation}
\label{eq:multiplicativity}
  \mathrm{ev}_{1_\alpha} = 1_{1_\alpha} \, ,
  \qquad
  \mathrm{ev}_{\widetilde{X} \otimes X} = \mathrm{ev}_X \circ \big(\mathrm{id}_{X^L} \otimes \mathrm{ev}_{\widetilde{X}} \otimes \mathrm{id}_X\big) \, ,
\end{equation}
and similarly for the coevaluations; this is equivalent to $\iota = \mathrm{id}$ and $\nu = \mathrm{id}$. In general, horizontal composition is only preserved up to the specified coherence isomorphisms: $(\widetilde{\varphi} \otimes \varphi)^L = \nu \circ (\varphi^L \otimes \widetilde{\varphi}^L) \circ \nu^{-1}$, while in the strict case it does so on the nose. For notational simplicity all statements below suppress the coherence isomorphisms; they can always be uniquely inserted.
\end{remark}

\begin{remark}
\label{rem:repackaging}
Definition~\ref{def:strictO2} can be repackaged as follows: the second datum is equivalent to an identity-on-objects dagger 2-functor $(-)^L \colon \cB \to \cB^{1\op,2\op}$ equipped with families $\mathrm{ev}_X, \mathrm{coev}_X$ exhibiting $X^L \dashv X$ which are extranatural with respect to $(-)^L$. Indeed, the extranaturality identities
\begin{equation}
  \mathrm{ev}_Y \circ (\mathrm{id}_{Y^L} \otimes \varphi) = \mathrm{ev}_X \circ (\varphi^L \otimes \mathrm{id}_X) \, ,
  \qquad
  (\varphi \otimes \mathrm{id}_{X^L}) \circ \mathrm{coev}_X = (\mathrm{id}_Y \otimes \varphi^L) \circ \mathrm{coev}_Y
\end{equation}
for $\varphi \colon X \Rightarrow Y$ each determine $\varphi^L$ uniquely via the Zorro moves, and are satisfied precisely when the functor acts by left mates. The adjunction data cannot be dispensed with, however: a bare identity-on-objects dagger 2-functor $\cB \to \cB^{1\op,2\op}$ bears no relation to duality, and even among mate-functors the choice of $\mathrm{ev}, \mathrm{coev}$ is genuine structure---by~\cite{Pen20}, in the one-object case, inequivalent unitary dual functors differ exactly by positive rescalings of this data.
\end{remark}

The first property we need is that the dagger turns left adjunction data into right adjunction data.

\begin{lemma}
\label{lem:abstractunitaryduality}
Let $\cB$ carry an $\Or(2)$-dagger structure. Then for every 1-morphism $X \colon \alpha \to \beta$ the 2-morphisms
\begin{equation}
  \widetilde{\mathrm{ev}}_X := (\mathrm{coev}_X)^\dg \colon X \otimes X^L \Longrightarrow 1_\beta \, ,
  \qquad
  \widetilde{\mathrm{coev}}_X := (\mathrm{ev}_X)^\dg \colon 1_\alpha \Longrightarrow X^L \otimes X
\end{equation}
exhibit $X^L$ also as a \emph{right} adjoint of $X$. In particular every 1-morphism has a two-sided adjoint.
\end{lemma}

\begin{proof}
Apply $\dg$ to the two Zorro identities for $(\mathrm{ev}_X, \mathrm{coev}_X)$. Since $\dg$ is the identity on 1-morphisms, reverses vertical and preserves horizontal composition, the dagger of, e.g., $(\mathrm{ev}_X \otimes \mathrm{id}_{X^L}) \circ (\mathrm{id}_{X^L} \otimes \mathrm{coev}_X) = 1_{X^L}$ is $(\mathrm{id}_{X^L} \otimes \widetilde{\mathrm{ev}}_X) \circ (\widetilde{\mathrm{coev}}_X \otimes \mathrm{id}_{X^L}) = 1_{X^L}$, which is a Zorro identity for the adjunction $X \dashv X^L$ with unit $\widetilde{\mathrm{coev}}_X$ and counit $\widetilde{\mathrm{ev}}_X$.
\end{proof}

With both adjunctions in hand, every 2-endomorphism has a right and a left trace.

\begin{definition}
\label{def:traces}
For a 2-endomorphism $f \colon X \Rightarrow X$ of a 1-morphism $X \colon \alpha \to \beta$, the \emph{right} and \emph{left traces} $\operatorname{tr}^R(f) \in \operatorname{End}(1_\alpha)$ and $\operatorname{tr}^L(f) \in \operatorname{End}(1_\beta)$ are
\begin{align}
  \operatorname{tr}^R(f) &:= \mathrm{ev}_X \circ (\mathrm{id}_{X^L} \otimes f) \circ \widetilde{\mathrm{coev}}_X \, , \\
  \operatorname{tr}^L(f) &:= \widetilde{\mathrm{ev}}_X \circ (f \otimes \mathrm{id}_{X^L}) \circ \mathrm{coev}_X \, ,
\end{align}
formed from the left adjunction data $\mathrm{ev}_X, \mathrm{coev}_X$ and its dagger $\widetilde{\mathrm{ev}}_X, \widetilde{\mathrm{coev}}_X$ of Lemma~\ref{lem:abstractunitaryduality}.
\end{definition}

A \emph{pivotal structure} on a bicategory $\cB$ with adjoints is a coherent identification of left and right adjoints. A convenient formulation is as a pseudonatural transformation $\delta \colon \mathrm{id} \Rightarrow (-)^{LL}$ whose object components are identity 1-morphisms; in the language of higher dagger categories this is an $\SO(2)$-dagger structure. We call such a structure \emph{unitary} if all 2-morphism components $\delta_X$ are unitary. The identification of left and right adjoints from Lemma~\ref{lem:abstractunitaryduality} provides exactly this:
\begin{proposition}
\label{prop:abstractpivotal}
An $\Or(2)$-dagger structure induces a canonical unitary pivotal structure on $\cB$: the 2-morphisms
\begin{equation}
  \delta_X := (\widetilde{\mathrm{ev}}_X \otimes \mathrm{id}_{X^{LL}}) \circ (\mathrm{id}_X \otimes \mathrm{coev}_{X^L}) \colon X \Longrightarrow X^{LL}
\end{equation}
are unitary (hence invertible) and assemble into a natural transformation $\mathrm{id}_\cB \Rightarrow (-)^{LL}$.
\end{proposition}

\begin{proof}
By Lemma~\ref{lem:abstractunitaryduality}, both $X$ (via the tilde data) and $X^{LL}$ (via the chosen data for $X^L$) are left adjoints of $X^L$; $\delta_X$ is the canonical comparison 2-morphism between them, invertible by uniqueness of adjoints, and naturality and monoidality follow from functoriality of $(-)^L$ and $\dg$. Unitarity $\delta_X^\dg = \delta_X^{-1}$ holds as well: in the one-object case this is~\cite[Prop.~3.9, Cor.~3.10]{Pen20} (see also~\cite[Thm.~2.8.1]{Ste23}), which shows that for a unitary dual functor $(-)^L$ on a unitary multitensor category the morphism $\delta_c = (\widetilde{\mathrm{ev}}_c \otimes \mathrm{id}_{c^{LL}}) \circ (\mathrm{id}_c \otimes \mathrm{coev}_{c^L})$ is automatically unitary; the same argument applies hom-category-wise in the bicategorical setting.
\end{proof}

Besides the dagger $\dg \colon \cB \to \cB^{2\op}$, which is part of the definition, an $\Or(2)$-dagger structure gives rise to a second dagger operation reversing the \emph{horizontal} direction. In the terminology of~\cite{FHJF24,CL25,Mue25,FMPS26}, the two reflections generate the subgroup $\bZ_2^b \times \bZ_2^t \leq \Or(2)$: the given dagger is $\dg^t$, reversing the direction of 2-morphisms (`$t$ for top'), and the derived one is $\dg^b$, reversing the direction of 1-morphisms (`$b$ for bottom').

\begin{lemma}
\label{lem:conjugation}
Define $(-)^* = \dg^b \colon \cB \to \cB^{1\op}$ to be the identity on objects, $X^* := X^L$ on 1-morphisms, and
\begin{equation}
  \varphi^* := (\varphi^L)^\dg = (\varphi^\dg)^L \colon X^* \Longrightarrow Y^*
  \qquad \text{for } \varphi \colon X \Rightarrow Y \, ,
\end{equation}
which is well defined by the commutation axiom of Definition~\ref{def:strictO2}. Then:
\begin{enumerate}
  \item $(-)^*$ is anti-linear, \emph{covariant} for vertical composition, $(\psi \circ \varphi)^* = \psi^* \circ \varphi^*$, and contravariant for horizontal composition, $(\widetilde{X} \otimes X)^* = X^* \otimes \widetilde{X}^*$ and $(\widetilde{\varphi} \otimes \varphi)^* = \varphi^* \otimes \widetilde{\varphi}^*$; that is, $(-)^*$ reverses only the direction of 1-morphisms.
  \item $\dg \circ (-)^* = (-)^* \circ \dg = (-)^L$.
  \item $((-)^*)^2 = (-)^{LL}$, so $(-)^*$ is involutive up to the canonical pivotal structure of Proposition~\ref{prop:abstractpivotal}: $\varphi^{**} = \delta_Y \circ \varphi \circ \delta_X^{-1}$ for $\varphi \colon X \Rightarrow Y$, and $X^{**} = X^{LL} \cong X$ via $\delta_X$.
\end{enumerate}
\end{lemma}

\begin{proof}
Direct computation from the axioms. The variances in (1) follow because $(-)^L$ reverses both compositions and $\dg$ reverses the vertical one again, e.g.\ $(\psi \circ \varphi)^* = ((\psi \circ \varphi)^L)^\dg = (\varphi^L \circ \psi^L)^\dg = \psi^* \circ \varphi^*$. For (2): $(\varphi^*)^\dg = ((\varphi^L)^\dg)^\dg = \varphi^L$ and $(\varphi^\dg)^* = ((\varphi^\dg)^\dg)^L = \varphi^L$, using $\dg^2 = \mathrm{id}$. For (3): $\varphi^{**} = (((\varphi^L)^\dg)^L)^\dg = (((\varphi^L)^L)^\dg)^\dg = \varphi^{LL}$, and the identification with $\delta$-conjugation is the naturality of $\delta$.
\end{proof}

\begin{remark}
\label{rem:O2reflections}
All these structures have natural interpretations in terms of the topological group $\Or(2)$:
\begin{itemize}
  \item $\dg \colon \cB \to \cB^{2\op}$ is the reflection across the \emph{horizontal} axis: it fixes objects and 1-morphisms and reverses the vertical direction, with $\dg^2 = \mathrm{id}$ strictly. It makes $\cB$ a $\bZ_2^t$-dagger bicategory;
  \item $(-)^* \colon \cB \to \cB^{1\op}$ is the reflection across the \emph{vertical} axis: it reverses the horizontal direction only, with $((-)^*)^2 = (-)^{LL}$ trivialized by $\delta$. It makes $\cB$ a $\bZ_2^b$-dagger bicategory;
  \item The composite $(-)^L = \dg \circ (-)^* \colon \cB \to \cB^{1\op,2\op}$ is the rotation by $\pi$; the trivialization $(-)^{LL} \cong \operatorname{id}$ of its square through the canonical pivotal structure of Proposition~\ref{prop:abstractpivotal} is the path corresponding to a full rotation. It is the extra structure of an $\SO(2)$-dagger bicategory.
\end{itemize}
\end{remark} 

\begin{remark}
\label{rem:unitarydualfunctors}
For $\cB$ with a single object $\alpha$, whose endomorphism category $\cC = \cB(\alpha,\alpha)$ is a unitary multitensor category (datum~(1) providing the dagger), Definition~\ref{def:strictO2} is verbatim the datum of a \emph{unitary dual functor} in the sense of~\cite{Pen20}: a dual functor $(-)^L$ on $\cC$ which is a dagger tensor functor with unitary coherence data. In the multi-object case (genuine bicategory), datum~(2) alone is the \emph{unitary adjoint functor} (UAF) of~\cite{CHFHS24}: a choice of adjunction data for each 1-morphism such that the induced left-mate functor $\cB \to \cB^{1\op,2\op}$ is dagger and its canonical tensorators are unitary. Note that unitarity of the tensorator is a genuine condition on the adjunction data and does \emph{not} follow from the commutation of $(-)^L$ with $\dg$ alone. Lemmas~\ref{lem:abstractunitaryduality} and~\ref{lem:conjugation} then reproduce the associated bi-involutive structure of~\cite{HP17}, with conjugation $\overline{X} = X^*$ and the daggered duality maps, and Proposition~\ref{prop:abstractpivotal} the canonical pivotal structure. The main theorem of~\cite{Pen20} classifies unitary dual functors on a fixed unitary multitensor category: they form a torsor over the group of groupoid homomorphisms from the universal grading groupoid $\mathcal{U}(\cC)$ to $\mathbb{R}_{>0}$, all inducing unitarily equivalent bi-involutive structures. For Definition~\ref{def:strictO2} this means: for a fixed dagger structure the choice of $(-)^L$ is genuinely non-unique, with the ambiguity controlled by positive rescalings of the duality data, so an $\Or(2)$-dagger structure is more data than a dagger structure---but only mildly so, and not at the level of the induced bi-involutive structure. 
\end{remark}

\section{Reflection positive defect TQFTs}
\label{sec:proposal}

In this section we define a dagger structure on the defect bordism category---equivalently, in its coherent reformulation, the $\bZ_2$-action of reflections on it together with preferred hermitian structures on all objects. We work in arbitrary dimension $n$ before specializing to $n = 1, 2$ in Section~\ref{sec:extractionsection}. We suppress standard details concerning collars and smooth structures on stratified bordisms. Throughout we follow~\cite{CRS18} for the stratified geometry and~\cite[\S 5]{Ste23} for the collar and dagger conventions in the closed case, which carry over unchanged.

We review the general definition of label sets for defects from~\cite[\S 2]{CRS18}, which we follow. A set of \emph{$n$-dimensional defect data}
\begin{equation}
  \bD = \big(D_n, D_{n-1}, \dots, D_0;\ f_{n-1}, f_{n-2}, \dots, f_0\big)
\end{equation}
consists of label sets $D_j$ for the $j$-dimensional strata for \emph{all} $j \in \{0, 1, \dots, n\}$, together with adjacency maps
\begin{equation}
  f_j \colon D_j \times \{\pm\} \longrightarrow \big[\mathrm{Sphere}^{\textrm{def}}_{n-j-1}(\partial^{j+1}\bD)\big]
\end{equation}
assigning to each label the isomorphism class of decorated stratified sphere constituting the link of the stratum. The definition is inductive in $n$~\cite[Definitions~2.1 and~2.4]{CRS18}: the allowed local configurations are iterated cones and cylinders over lower-dimensional configurations, the link spheres are endomorphisms of $\varnothing$ in the lower-dimensional decorated bordism category, and $\partial\bD$ denotes the $(n-1)$-dimensional defect data obtained by forgetting $D_0$ and $f_0$. The adjacency maps satisfy the duality condition
\begin{equation}
\label{eq:CRSduality}
  f_j(\varphi, -) = f_j(\varphi, +)^{\mathrm{rev}} \, ,
\end{equation}
where $(-)^{\mathrm{rev}}$ reverses the orientations of \emph{all} strata~\cite[eq.~(2.21)]{CRS18}. The symmetric monoidal category $\Bord_n^{\textrm{def}}(\bD)$ has as objects closed stratified oriented $(n-1)$-manifolds decorated by $\partial\bD$, and as morphisms compact decorated stratified oriented bordisms---the strata of all dimensions are oriented and labelled, with links prescribed by the $f_j$---taken up to decoration-preserving isomorphism relative to the boundary. An $n$-dimensional \emph{defect TQFT} is a symmetric monoidal functor $\cZ \colon \Bord_n^{\textrm{def}}(\bD) \to \Vect_\bC$. As usual, other symmetric monoidal target categories (such as super vector spaces) may also be considered. 

For $n = 2$ we spell this out in slightly more detail, fixing the conventions used in the proofs later on. Following~\cite{Car16} we take $D_0 = \varnothing$, i.e.\ there are no interior 0-strata; by the $D_0$-completion of~\cite[\S 2.4]{CRS18} this is no loss of generality, as point defect labels can be identified with states computed by the TQFT---which is exactly how the 2-morphism spaces of $\cT_\cZ$ arise in Section~\ref{sec:defectTQFT}. A set of 2-dimensional defect data with $D_0 = \varnothing$ then amounts to a tuple $\bD = (D_2, D_1, s, t)$ consisting of a set $D_2$ of bulk labels, a set $D_1$ of line defect labels, and source and target maps $s,t \colon D_1 \to D_2$: the link sphere of a 1-stratum is a 0-sphere whose two points are decorated by the adjacent bulk labels, so $f_1(X,+)$ encodes the pair $(s(X), t(X))$, and the duality condition \eqref{eq:CRSduality} determines $f_1(X,-)$.

We fix the following orientation conventions, illustrated in Figure~\ref{fig:coorientation}. All 2-strata carry the ambient orientation, and all 1-strata are oriented. Let $\ell$ be a 1-stratum, $p \in \ell$, and let $v$ be a positively oriented tangent vector of $\ell$ at $p$. The \emph{coorientation} of $\ell$ at $p$ is the unique normal direction $\nu$ such that $(v, \nu)$ is a positively oriented frame for the ambient orientation. In the standard plane with its counterclockwise orientation and $\ell$ oriented upwards, $\nu$ points to the left. The coorientation distinguishes the two adjacent 2-strata: we call the one which $\nu$ points out of the \emph{source} and the one which $\nu$ points into the \emph{target}. Concretely, the source lies to the right of $\ell$ and the target to its left when traversing $\ell$ along its orientation. A 1-stratum labelled by $X \in D_1$ is then required to have its source 2-stratum labelled by $s(X)$ and its target 2-stratum by $t(X)$. Note that $\nu$ is unchanged if both the ambient orientation and the orientation of $\ell$ are reversed, since reversing the ambient orientation and reversing $v$ each flip $\nu$, and the two reversals cancel; this elementary observation underlies the well-definedness of total orientation reversal in Section~\ref{sec:involutivedata}.

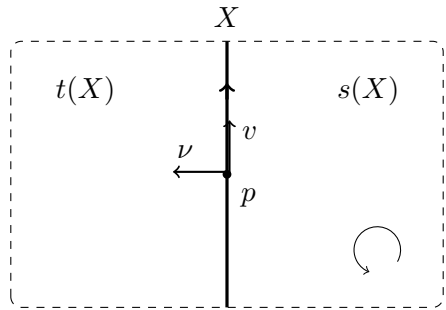
\begin{figure}[ht]
\centering
\begin{tikzpicture}[scale=1.1]
  \draw[dashed, rounded corners] (-2.6,-1.6) rectangle (2.6,1.6);
  \draw[very thick] (0,-1.6) -- (0,1.6);
  \draw[very thick,->] (0,0.9) -- (0,1.1);
  \node[above] at (0,1.65) {$X$};
  \fill (0,0) circle (1.5pt);
  \node[below right] at (0.05,-0.05) {$p$};
  \draw[->,thick] (0.03,0) -- (0.03,0.65); \node[right] at (0.06,0.5) {$v$};
  \draw[->,thick] (0,0.03) -- (-0.65,0.03); \node[above] at (-0.5,0.08) {$\nu$};
  \node at (-1.7,1.0) {$t(X)$};
  \node at (1.7,1.0) {$s(X)$};
  \draw[->] (2.05,-1.05) arc (-30:250:0.28);
\end{tikzpicture}
\caption{Coorientation convention for interior 1-strata. The 1-stratum $\ell$ is oriented by $v$; the ambient orientation (counterclockwise, indicated in the corner) rotates $v$ by $+\tfrac{\pi}{2}$ into the normal $\nu$, so that $(v,\nu)$ is a positively oriented frame. The 2-stratum which $\nu$ points out of (here: to the right of $\ell$) is the source, the one it points into is the target.}
\label{fig:coorientation}
\end{figure}

The symmetric monoidal category $\Bord_2^{\textrm{def}}(\bD)$ has as objects closed oriented 1-manifolds $E$ with finitely many marked points, each marked point labelled by a pair $(X,\varepsilon) \in D_1 \times \{\pm\}$ recording the defect line which will cross transversally and the sign of the crossing, and complementary intervals labelled by $D_2$, subject to the matching conditions imposed by $s$ and $t$. Morphisms are (equivalence classes of) compact oriented stratified bordisms with oriented 1-strata labelled by $D_1$, 2-strata labelled by $D_2$, compatibly with the conventions above and with the boundary decorations; there are no interior 0-strata.

We explain the boundary conventions for cobordisms in more detail as they will be important for the computations below. We model an object $E$ by the germ of a decorated cylinder $E \times (-1,1)$ with ambient orientation $\mathrm{or}(E) \times (\uparrow)$, where the 1-stratum through a marked point $p$ is oriented upward if $\varepsilon_p = +$ and downward if $\varepsilon_p = -$. The matching conditions then read: for $\varepsilon_p = +$ the interval to the right of $p$ (with respect to the orientation of $E$) is labelled $s(X_p)$ and the interval to the left $t(X_p)$; for $\varepsilon_p = -$ the two are exchanged. Incoming boundary components of a bordism are read with the inward pointing normal as positive vertical direction, outgoing components with the outward pointing normal; with these conventions the cylinder $E \times [0,1]$ with product decoration is the identity on $E$. The conventions are illustrated in Figure~\ref{fig:conventions}.

\begin{figure}[ht]
\centering
\begin{tikzpicture}[baseline, scale=0.9]
  \draw[->] (-1.8,0) -- (1.8,0);
  \node[right] at (1.8,0) {$E$};
  \draw[very thick, ->] (0,-1.2) -- (0,1.2);
  \node[above] at (0,1.2) {$X$};
  \node at (-1.1,0.7) {$t(X)$};
  \node at (1.1,0.7) {$s(X)$};
  \fill (0,0) circle (1.6pt);
  \node[below right] at (0.05,-0.1) {$(X,+)$};
\end{tikzpicture}
\hspace{1.8cm}
\begin{tikzpicture}[baseline, scale=0.9]
  \draw[->] (-1.8,0) -- (1.8,0);
  \node[right] at (1.8,0) {$E$};
  \draw[very thick, <-] (0,-1.2) -- (0,1.2);
  \node[above] at (0,1.2) {$X$};
  \node at (-1.1,0.7) {$s(X)$};
  \node at (1.1,0.7) {$t(X)$};
  \fill (0,0) circle (1.6pt);
  \node[below right] at (0.05,-0.1) {$(X,-)$};
\end{tikzpicture}
\caption{Germ presentation of a marked point on an object $E$ (drawn locally; $E$ oriented to the right, vertical direction $\uparrow$, ambient orientation $\mathrm{or}(E) \times (\uparrow)$). For $\varepsilon_p = +$ the defect line is oriented upward and $s(X)$ sits to the right of $p$; for $\varepsilon_p = -$ both are reversed.}
\label{fig:conventions}
\end{figure}
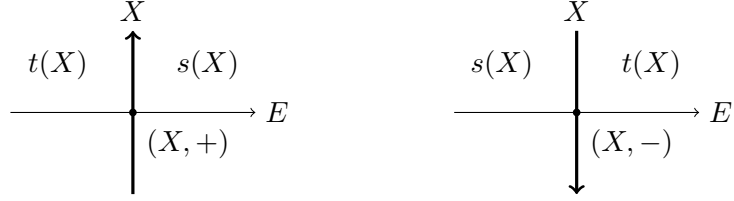

\subsection{The reflection action}
\label{sec:involutivedata}

\begin{definition}
\label{def:reflectionfunctor}
The \emph{total orientation reversal functor}
\begin{equation}
  R \colon \Bord_n^{\textrm{def}}(\bD) \longrightarrow \Bord_n^{\textrm{def}}(\bD)
\end{equation}
sends a decorated stratified manifold (or object) to the same underlying stratified manifold with the orientations of \emph{all} strata reversed and all labels unchanged.
\end{definition}

The functor $R$ is the operation $(-)^{\mathrm{rev}}$ of~\cite{CRS18} applied to objects and bordisms. It is well-defined for \emph{every} set of defect data, with no further compatibility required: total orientation reversal modifies the decorated link of a stratum by total orientation reversal in one dimension lower, and this is admissible precisely by the duality condition \eqref{eq:CRSduality}, which is hardwired into the definition of defect data. For $n = 2$ this specializes to the statement that reversing both the ambient orientation and the orientation of a 1-stratum preserves its coorientation, so the source and target conditions for the unchanged line labels continue to hold. Moreover $R^2 = \mathrm{id}$ strictly, and $R$ is symmetric monoidal. We write $\overline{(-)} \colon \Vect_\bC \to \Vect_\bC$ for the complex conjugation functor.

The value of $R$ on an object is a canonical model for its dual, as follows from the standard bending of cylinders argument, which we record as:

\begin{lemma}
\label{lem:Rvsdual}
\leavevmode
\begin{enumerate}
  \item Every object $E \in \Bord_n^{\textrm{def}}(\bD)$ is dualizable, with $E^\vee$ given by $E$ with the orientations of all strata reversed and all labels unchanged; evaluation and coevaluation are the bent cylinders over $E$.
  \item $R(E) = E^\vee$ for every object $E$, i.e.\ $R$ agrees with $(-)^\vee$ on objects.
\end{enumerate}
\end{lemma}

The bent cylinder is shown for $n = 2$ in Figure~\ref{fig:bentcylinder}.

\begin{figure}[ht]
\centering
\begin{tikzpicture}[scale=0.85]
  \draw (-2,0) -- (-2,1) arc (180:0:2) -- (2,0);
  \draw (-1,0) -- (-1,1) arc (180:0:1) -- (1,0);
  \draw[very thick] (-2,0) -- (-1,0);
  \draw[very thick] (1,0) -- (2,0);
  \draw[thick] (-1.5,0) -- (-1.5,1) arc (180:0:1.5) -- (1.5,0);
  \draw[thick,->] (1.5,0.25) -- (1.5,0.75);
  \draw[thick,->] (-1.5,0.75) -- (-1.5,0.25);
  \fill (-1.5,0) circle (1.6pt);
  \fill (1.5,0) circle (1.6pt);
  \node[below] at (-1.5,-0.05) {$(X,-)$};
  \node[below] at (1.5,-0.05) {$(X,+)$};
  \node[below] at (-1.5,-0.6) {$R(E) = E^\vee$};
  \node[below] at (1.5,-0.6) {$E$};
\end{tikzpicture}
\caption{Schematic sketch of the evaluation bent cylinder $h_E \colon E^\vee \sqcup E \to \varnothing$ over an object $E$ with one marked point $(X,+)$, both ends incoming. The defect line enters the $E$-end inward and exits the other end outward, flipping the sign.}
\label{fig:bentcylinder}
\end{figure}
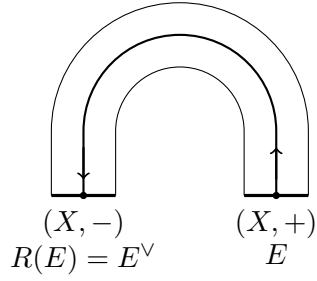

\begin{remark}
\label{rem:hermitianstructures}
In the language of~\cite[Appendix B]{FH16} and~\cite{SS23,Ste23}, the symmetric monoidal $\bZ_2$-action $R$ together with the duals $(-)^\vee$ combine to produce an anti-involution $E \longmapsto R(E^\vee)$ on $\Bord_n^{\textrm{def}}(\bD)$. The hermitian structures of~\cite{SS23,Ste23} are exactly the data trivialising this anti-involution on objects, i.e.\ isomorphisms $E \xrightarrow{\sim} R(E^\vee)$; with the specific realization $E^\vee = R(E)$ of Lemma~\ref{lem:Rvsdual}, these can be chosen to be trivial ($R(E^\vee) = R(R(E)) = E$ strictly, so the identity isomorphisms suffice), and the nontrivial content is the induced action on morphisms, which reverses their direction. For other tangential structures the choice of dual realization is more subtle and the hermitian structures are genuinely nontrivial; we refer to~\cite{Ste23} for a detailed discussion. As explained in~\cite{SS23,Ste23}, a compatible collection of hermitian structures canonically turns $\Bord_n^{\textrm{def}}(\bD)$ into a symmetric monoidal dagger category; we spell out the resulting dagger structure in the next subsection.
\end{remark}

\begin{definition}
\label{def:hermitiandefectTQFT}
A \emph{reflection defect TQFT} is a $\bZ_2$-equivariant symmetric monoidal functor $\cZ \colon \Bord_n^{\textrm{def}}(\bD) \to \Vect_\bC$. The equivariance datum is the additional choice of a symmetric monoidal natural isomorphism
\begin{equation}
  \rho \colon \cZ \circ R \Longrightarrow \overline{\cZ}
\end{equation}
satisfying the equivariance coherence condition $\overline{\rho} \circ (\rho \circ R) = \mathrm{id}_{\cZ}$.
\end{definition}

\begin{definition}
\label{def:RPproperty}
For an object $E \in \Bord_n^{\textrm{def}}(\bD)$ denote by $h_E \colon R(E) \sqcup E \to \varnothing$ the evaluation bent cylinder, using $R(E) = E^\vee$ from Lemma~\ref{lem:Rvsdual}, and define
\begin{equation}
  b_E := \cZ(h_E) \circ \big(\rho_E^{-1} \otimes \mathrm{id}\big) \colon \overline{\cZ(E)} \otimes \cZ(E) \longrightarrow \bC \, ,
\end{equation}
where $\rho_E^{-1} \colon \overline{\cZ(E)} \to \cZ(R(E))$ is the inverse of the structure isomorphism. A reflection defect TQFT $(\cZ, \rho)$ is \emph{reflection positive} if $b_E$ is positive definite for every object $E$.
\end{definition}

\begin{remark}
\label{rem:hermiticityofb}
That $b_E$ is hermitian, i.e.\ $\overline{b_E(v,w)} = b_E(w,v)$, follows from the fixed point coherence condition on $\rho$, as in the closed case~\cite{FH16,Ste23}.
\end{remark}

\subsection{The dagger category of defect bordisms}
\label{sec:daggerdefbordcat}
As explained in Remark~\ref{rem:hermitianstructures}, the trivial hermitian structures provided by Lemma~\ref{lem:Rvsdual} equip $\Bord_n^{\textrm{def}}(\bD)$ with the structure of a symmetric monoidal dagger category, which we record explicitly.
\begin{definition}
\label{def:daggerstructure}
For a defect bordism $\Sigma \colon E \to F$ in $\Bord_n^{\textrm{def}}(\bD)$ define
\begin{equation}
  \Sigma^\dg \colon F \longrightarrow E
\end{equation}
to be the bordism with the same underlying stratified manifold, the orientations of all strata reversed, all labels unchanged, regarded as a bordism from $F$ to $E$ via the canonical identification $\partial R(\Sigma) \cong R(\partial \Sigma)$ together with the exchange of incoming and outgoing boundary. This turns $\Bord_n^{\textrm{def}}(\bD)$ into a symmetric monoidal dagger category. 
\end{definition}

By the general theory of~\cite[Thm.~2.3.41, \S\S2.4, 5.2]{Ste23}, Definition~\ref{def:hermitiandefectTQFT} together with the reflection positivity condition of Definition~\ref{def:RPproperty} is equivalent to the following: the equivariance datum $\rho$ makes $\cZ$ into an anti-involutive functor from $(\Bord_n^{\textrm{def}}(\bD), R)$ to $(\Vect_\bC, \overline{(-)})$, and positivity of $b_E$ for all $E$ says precisely that $\cZ$ maps the canonical positivity structure on $\Bord_n^{\textrm{def}}(\bD)$ (given by the trivial hermitian pairings of Remark~\ref{rem:hermitianstructures}) to the positive-definite positivity structure on $\Hilb$; the dagger functor formulation packages both pieces of data at once.

\begin{definition}
\label{def:RPdefectTQFT}
A \emph{reflection positive defect TQFT} is a symmetric monoidal dagger functor
\begin{equation}
  \cZ \colon \Bord_n^{\textrm{def}}(\bD) \longrightarrow \Hilb \, .
\end{equation}
\end{definition}

\begin{remark}
\label{rem:positivityforfree}
In this formulation positivity is wired in through the target: every state space $\cZ(E)$ is a Hilbert space. Since the 2-morphism spaces of the bicategory extracted in Section~\ref{sec:extractionsection} for $n = 2$ are state spaces of decorated circles, they will be Hilbert spaces in a reflection positive defect TQFT.
\end{remark}

\section{Higher categories of defects in reflection TQFTs}
\label{sec:extractionsection}

We now extract the $\Or(2)$-dagger structure from a 2-dimensional reflection defect TQFT, warming up with the one-dimensional case (Section~\ref{sec:dim1}) before recalling the two-dimensional defect bicategory $\cT_\cZ$ (Section~\ref{sec:defectTQFT}) and proving the main theorem in dimension two (Section~\ref{sec:dim2}).

\subsection{Dimension one}
\label{sec:dim1}

We begin with $n = 1$. Fix one-dimensional defect data, a set $D_1$ of bulk labels; as in the $D_0$-completion of~\cite[\S 2.4]{CRS18} adopted in Section~\ref{sec:proposal}, interfaces are computed by the theory. Let $\cZ \colon \Bord_1^{\textrm{def}}(\bD)\to \Vect_\bC$ be a 1-dimensional defect TQFT. For $\alpha, \beta \in D_1$ the link of a putative interface is the decorated $0$-sphere $S^0_{\alpha,\beta}$ (a negatively oriented point labelled $\alpha$ and a positively oriented point labelled $\beta$), an object of $\Bord_1^{\textrm{def}}(\bD)$, and
\begin{equation}
	\operatorname{Hom}(\alpha,\beta) := \cZ\big(S^0_{\alpha,\beta}\big) \, .
\end{equation}
 Writing $V_\alpha := \cZ(\mathrm{pt}_\alpha)$ one has $\cZ(S^0_{\alpha,\beta}) \cong \operatorname{Hom}_\bC(V_\alpha, V_\beta)$, so the extracted category $\cT_\cZ$ has the bulk state spaces as objects and \emph{all} linear maps as morphisms, with composition induced by the interval bordism.

With $E := S^0_{\alpha,\beta}$, total orientation reversal gives $S^0_{\beta,\alpha} = R(E)$ (Lemma~\ref{lem:Rvsdual}), and the equivariance datum $\rho_E \colon \cZ(R(E)) \xrightarrow{\;\sim\;} \overline{\cZ(E)}$ \emph{is} the dagger---the anti-linear isomorphism
\begin{equation}
	\dg \colon \operatorname{Hom}(\alpha,\beta) \longrightarrow \operatorname{Hom}(\beta,\alpha), \qquad \varphi^\dg := \rho_E^{-1}(\overline{\varphi}) \, ,
\end{equation}
i.e.\ $\operatorname{Hom}(\beta,\alpha) \cong \overline{\operatorname{Hom}(\alpha,\beta)}$. Involutivity $(\varphi^\dg)^\dg = \varphi$ is immediate from the coherence $\overline{\rho} \circ (\rho \circ R) = \mathrm{id}$ (Definition~\ref{def:hermitiandefectTQFT}) and functoriality from naturality of $\rho$. Equivalently, $\varphi^\dg$ is the adjoint for the bent-interval pairing, i.e.\ $\cZ(h)(\varphi^\dg \otimes \psi) = b_E(\varphi, \psi)$ through the cup $h \colon S^0_{\beta,\alpha} \sqcup S^0_{\alpha,\beta} \to \varnothing$ (Figure~\ref{fig:bentinterval}), with $b_E$ the hermitian form of Definition~\ref{def:RPproperty}.

\begin{figure}[ht]
	\centering
	\begin{tikzpicture}[scale=0.95]
		\draw[thick] (-1.5,0.8) arc (180:360:1.5);
		\draw[thick] (-0.5,0.8) arc (180:360:0.5);
		\fill (-1.5,0.8) circle (1.6pt);
		\fill (-0.5,0.8) circle (1.6pt);
		\fill (0.5,0.8) circle (1.6pt);
		\fill (1.5,0.8) circle (1.6pt);
		\node[left] at (-1.05,-0.55) {$\alpha$};
		\node[right] at (-0.02,-0.05) {$\beta$};
		\node[below] at (-1.0,0.95) {$S^0_{\beta,\alpha}$};
		\node[below] at (1.0,0.95) {$S^0_{\alpha,\beta}$};
	\end{tikzpicture}
	\caption{The evaluation bent interval $h \colon S^0_{\beta,\alpha} \sqcup S^0_{\alpha,\beta} \to \varnothing$ in dimension one: the elementary cup pairing $\operatorname{Hom}(\beta,\alpha) \otimes \operatorname{Hom}(\alpha,\beta) \to \bC$, with the two arcs coloured by the bulk labels $\alpha$ and $\beta$. It is the one-dimensional analogue of the bent cylinder of Figure~\ref{fig:bentcylinder}.}
	\label{fig:bentinterval}
\end{figure}
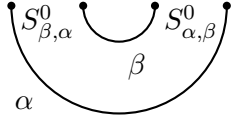

\begin{proposition}
	\label{prop:dim1}
	For a 1-dimensional reflection defect TQFT $(\cZ, \rho)$ the category $\cT_\cZ$ is a $\bC$-linear dagger category whose dagger is complex anti-linear.
\end{proposition}
 For a single bulk label this recovers the classification of one-dimensional theories with reflection structure by hermitian vector spaces~\cite{MS23}; reflection positivity singles out the positive definite ones (Section~\ref{sec:positivity}).

\subsection{The defect bicategory \texorpdfstring{$\cT_\cZ$}{TZ} in dimension 2}
\label{sec:defectTQFT}

We recall the bicategory $\cT_\cZ$ that a $2$-dimensional defect TQFT $\cZ \colon \Bord_2^{\textrm{def}}(\bD) \to \Vect_\bC$ determines, due to~\cite{DKR11} and reviewed in detail in~\cite[\S 2.3]{Car16}. Our conventions for $n = 2$ are those of~\cite{Car16}, set up in Section~\ref{sec:proposal}. It is convenient to extend the source and target maps to signed labels by $s(x,+) = s(x)$, $t(x,+) = t(x)$ and $s(x,-) = t(x)$, $t(x,-) = s(x)$.

\emph{Objects and 1-morphisms.} The objects of $\cT_\cZ$ are the bulk labels $D_2$. A 1-morphism $\alpha \to \beta$ is a composable list $X = ((x_1,\varepsilon_1), \dots, (x_n,\varepsilon_n))$ with $n \geq 0$, $s(x_n,\varepsilon_n) = \alpha$, $t(x_1,\varepsilon_1) = \beta$ and $s(x_i,\varepsilon_i) = t(x_{i+1},\varepsilon_{i+1})$, representing parallel defect lines; no fusion is performed at the level of labels. Horizontal composition is concatenation of lists, and the unit $1_\alpha$ is the empty list.

\emph{2-morphisms.} For 1-morphisms $X, Y \colon \alpha \to \beta$ let $E_{X,Y}$ denote the defect circle of~\cite[eq.~(2.16)]{Car16}: the points $(y_1,\nu_1)$, \dots, $(y_m,\nu_m)$ of $Y$ sit on the upper half, the points $(x_1,-\varepsilon_1)$, \dots, $(x_n,-\varepsilon_n)$ of $X$ with reversed signs on the lower half, and the two remaining intervals are labelled $\beta$ (left) and $\alpha$ (right). Then
\begin{equation}
  \operatorname{Hom}(X,Y) := \cZ(E_{X,Y}) \, .
\end{equation}
The defect circle is shown in Figure~\ref{fig:EXY}. The interpretation is that the TQFT computes the label set of junction points itself, by cutting out a hole around a putative junction.

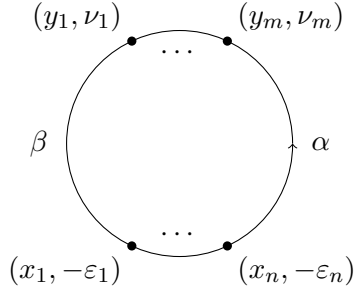
\begin{figure}[ht]
\centering
\begin{tikzpicture}[scale=1.15]
  \draw[->] (1.3,0) arc (0:360:1.3);
  \fill (115:1.3) circle (1.5pt); \node[above left] at (115:1.3) {$(y_1,\nu_1)$};
  \fill (65:1.3) circle (1.5pt); \node[above right] at (65:1.3) {$(y_m,\nu_m)$};
  \node at (90:1.05) {$\cdots$};
  \fill (245:1.3) circle (1.5pt); \node[below left] at (245:1.3) {$(x_1,-\varepsilon_1)$};
  \fill (295:1.3) circle (1.5pt); \node[below right] at (295:1.3) {$(x_n,-\varepsilon_n)$};
  \node at (270:1.05) {$\cdots$};
  \node[left] at (180:1.4) {$\beta$};
  \node[right] at (0:1.4) {$\alpha$};
\end{tikzpicture}
\caption{The defect circle $E_{X,Y}$ for $X, Y \colon \alpha \to \beta$, with $\operatorname{Hom}(X,Y) = \cZ(E_{X,Y})$: the points of $Y$ on the upper half, those of $X$ with reversed signs on the lower half.}
\label{fig:EXY}
\end{figure} Vertical composition is $\Psi \circ \Phi := \cZ(P_{X,Y,Z})(\Psi \otimes \Phi)$, where $P_{X,Y,Z} \colon E_{Y,Z} \sqcup E_{X,Y} \to E_{X,Z}$ is the pair of pants in which the $Y$-points of the two ingoing circles are joined by parallel defect lines in the interior; the unit is $1_X = \cZ(D_X)(1)$ for the disk $D_X \colon \varnothing \to E_{X,X}$ with parallel defect lines joining the $X$-points on the lower half to those on the upper half. Horizontal composition of 2-morphisms uses a second, differently decorated pair of pants. Associativity and unitality follow from isotopy invariance and functoriality.

In~\cite{DKR11} it is shown that $\cT_\cZ$ is a pivotal bicategory. The adjoint of a 1-morphism is orientation reversal of the defect lines, implemented on lists as
\begin{equation}
  X^L := \big((x_n, -\varepsilon_n), \dots, (x_1, -\varepsilon_1)\big) \colon \beta \longrightarrow \alpha \, ,
\end{equation}
which is simultaneously left and right adjoint to $X$; note $\operatorname{Hom}(X,Y) = \operatorname{Hom}(1_\beta, Y \otimes X^L)$. The left adjunction is exhibited by $\mathrm{ev}_X \colon X^L \otimes X \Rightarrow 1_\alpha$ and $\mathrm{coev}_X \colon 1_\beta \Rightarrow X \otimes X^L$, given by $\cZ$ of ``rainbow'' disks with nested arcs~\cite[eqs.~(2.18), (2.19)]{Car16}; the right adjunction data $\widetilde{\mathrm{ev}}_X, \widetilde{\mathrm{coev}}_X$ is obtained by reversing all orientations and orders. The Zorro moves are isotopies. Note that the rainbow data is normalized and multiplicative by construction: the rainbow disk of a concatenated list is the nesting of the individual rainbows (an isotopy identifies it with the corresponding composite), and the disk for the empty list is the identity 2-morphism.

\subsection{Reflection structure}
\label{sec:dim2}
Let $(\cZ, \rho)$ be a 2-dimensional reflection defect topological quantum field theory in the sense of Definition~\ref{def:hermitiandefectTQFT}, with defect bicategory $\cT_\cZ$ (Section~\ref{sec:defectTQFT}) built from the underlying $\Vect_\bC$-valued functor. Our main result is the following.

\begin{theorem}
\label{thm:extraction}
Let $(\cZ, \rho)$ be a reflection defect topological quantum field theory. Then the bicategory $\cT_\cZ$ of~\cite{DKR11,Car16} carries an $\Or(2)$-dagger structure in the sense of Definition~\ref{def:strictO2}, with $\dg$ induced by the reflection equivariance and $(-)^L$ induced by rotation of the defect lines by $\pi$.
\end{theorem}

If in addition $(\cZ,\rho)$ is reflection positive, this $\Or(2)$-dagger structure is positive (the forms $b_E$ are positive definite); we treat that in Section~\ref{sec:positivity}. The proof of Theorem~\ref{thm:extraction} occupies the remainder of this section.

We start by constructing the dagger, which is again induced by the reflection isomorphism $\rho$.

\paragraph{Part 1: construction of the dagger.}

We first rigidify the decorated circles, so that the diffeomorphisms used below are canonical.

\begin{remark}
\label{rem:canonicityu}
We fix the defect circles once and for all, following~\cite{Car16,DKR11}: the marked points sit at prescribed standard angular positions, chosen to be symmetric under the reflection across the real axis, and the basepoint $-1 \in S^1$ carries no marked point, ruling out identifications by nontrivial rotations. All the decorated circles below are presented in this standard form, so the diffeomorphisms and gluings that follow are canonical---unambiguous even when the circle has a cyclic label symmetry (e.g.\ all labels equal).
\end{remark}

\emph{Reflection of defect circles.}

\begin{lemma}
\label{lem:EXYdual}
For all 1-morphisms $X, Y \colon \alpha \to \beta$ there is an isomorphism
\begin{equation}
  u_{X,Y} \colon R(E_{X,Y}) = E_{X,Y}^\vee \xrightarrow{\;\sim\;} E_{Y,X}
\end{equation}
in $\Bord_2^{\textrm{def}}(\bD)$, induced by a decoration-preserving, orientation-preserving diffeomorphism.
\end{lemma}

\begin{proof}
By Lemma~\ref{lem:Rvsdual}, $R(E_{X,Y})$ is $E_{X,Y}$ with the circle orientation and all coorientations reversed and the labels unchanged, so the $Y$-points become $(y_j, -\nu_j)$ and the $X$-points become $(x_i, \varepsilon_i)$. The orientation-reversing reflection of the circle across the horizontal axis through the basepoint, $z \mapsto \bar z$, restores the standard orientation: it fixes the basepoint $-1$ and the intervals $\beta$ (left) and $\alpha$ (right) and exchanges the upper and lower halves. It therefore carries the $X$-points to the upper half and the $Y$-points to the lower half, producing the decorated circle with $X$-points $(x_i,\varepsilon_i)$ above and $Y$-points $(y_j,-\nu_j)$ below---which is $E_{Y,X}$ (Figure~\ref{fig:reflectEXY}). We take $u_{X,Y}$ to be this reflection, which is canonical by Remark~\ref{rem:canonicityu}.
\end{proof}
\begin{figure}[ht]
\centering
\begin{tikzpicture}[scale=1.0, >=stealth, baseline=(current bounding box.center)]
  \draw[dashed, gray] (-1.65,0) -- (1.65,0);
  \draw[thick] (0,0) circle (1.2);
  \draw[thick,->] (-7.5:1.2) arc (-7.5:7.5:1.2);
  \fill (0,1.2) circle (1.4pt); \node[above] at (0,1.27) {$(y,\nu)$};
  \fill (0,-1.2) circle (1.4pt); \node[below] at (0,-1.27) {$(x,-\varepsilon)$};
  \node at (-0.78,0) {$\beta$};
  \node at (0.78,0) {$\alpha$};
  \draw[fill=white] (-1.2,0) circle (1.7pt);
  \node[left] at (-1.32,0) {\small $-1$};
  \node at (0,-2.0) {$E_{X,Y}$};
\end{tikzpicture}
\qquad $\xrightarrow{\ u_{X,Y}\ }$ \qquad
\begin{tikzpicture}[scale=1.0, >=stealth, baseline=(current bounding box.center)]
  \draw[dashed, gray] (-1.65,0) -- (1.65,0);
  \draw[thick] (0,0) circle (1.2);
  \draw[thick,->] (-7.5:1.2) arc (-7.5:7.5:1.2);
  \fill (0,1.2) circle (1.4pt); \node[above] at (0,1.27) {$(x,\varepsilon)$};
  \fill (0,-1.2) circle (1.4pt); \node[below] at (0,-1.27) {$(y,-\nu)$};
  \node at (-0.78,0) {$\beta$};
  \node at (0.78,0) {$\alpha$};
  \draw[fill=white] (-1.2,0) circle (1.7pt);
  \node[left] at (-1.32,0) {\small $-1$};
  \node at (0,-2.0) {$E_{Y,X}$};
\end{tikzpicture}
\caption{The isomorphism $u_{X,Y}$ of Lemma~\ref{lem:EXYdual} as the reflection of the defect circle across the horizontal axis (dashed). Total orientation reversal followed by $z \mapsto \bar z$ fixes the basepoint $-1$ and the intervals $\beta$ (left), $\alpha$ (right), exchanges the upper and lower halves, and flips the coorientation signs, carrying $E_{X,Y}$ ($Y$ above, $X$ below) to $E_{Y,X}$ ($X$ above, $Y$ below). A single $Y$- and $X$-point are drawn; general lists are reflected pointwise.}
\label{fig:reflectEXY}
\end{figure}
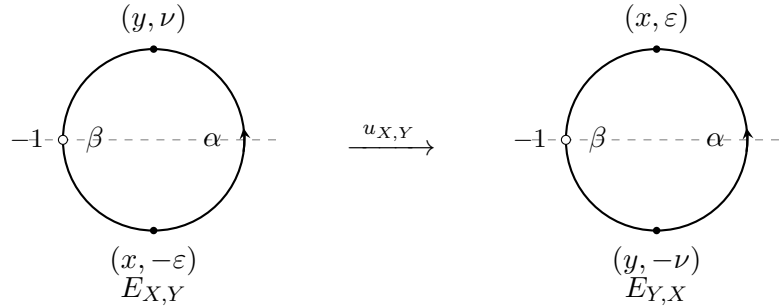

We abbreviate $E := E_{X,Y}$, $E' := E_{Y,X}$, and write $h := h_{E} \circ (u_{X,Y}^{-1} \sqcup \mathrm{id}) \colon E' \sqcup E \to \varnothing$ for the evaluation bent cylinder presented on the standard circles.

\emph{Definition of the dagger.} As in dimension one (Section~\ref{sec:dim1}), the dagger is the reflection isomorphism furnished by $\rho$.

\begin{definition}
\label{def:daggeronTZ}
By Lemma~\ref{lem:EXYdual} the reflected circle $R(E)$ is identified with $E' = E_{Y,X}$ through $u_{X,Y}$; composing $\cZ(u_{X,Y})$ with $\rho_E^{-1} \colon \overline{\cZ(E)} \xrightarrow{\ \sim\ } \cZ(R(E))$ gives, for $\varphi \in \operatorname{Hom}(X,Y) = \cZ(E)$,
\begin{equation}
\label{eq:daggerdef}
  \varphi^\dg := \cZ(u_{X,Y})\big(\rho_E^{-1}(\overline{\varphi})\big) \;\in\; \cZ(E') = \operatorname{Hom}(Y,X) \, .
\end{equation}
\end{definition}

This is anti-linear and exhibits $\operatorname{Hom}(Y,X) = \cZ(E') \cong \cZ(R(E)) \cong \overline{\cZ(E)} = \overline{\operatorname{Hom}(X,Y)}$. Geometrically it is total orientation reversal of a decorated disk, re-read as a state on the reflected circle, as illustrated in Figure~\ref{fig:daggerdisk}.

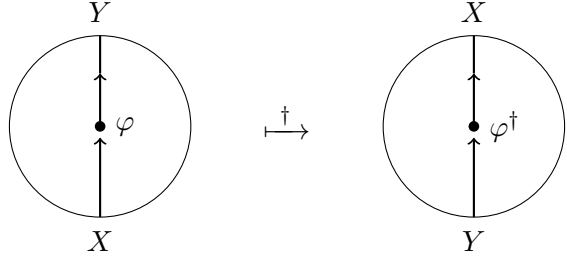
\begin{figure}[ht]
\centering
\begin{tikzpicture}[scale=1.0, baseline={(0,-0.1)}]
  \draw (0,0) circle (1.2);
  \draw[thick,->] (0,-1.2) -- (0,-0.15);
  \draw[thick,->] (0,0) -- (0,0.7);
  \draw[thick] (0,0.7) -- (0,1.2);
  \fill (0,0) circle (2pt); \node[right] at (0.07,0) {$\varphi$};
  \node[below] at (0,-1.25) {$X$};
  \node[above] at (0,1.25) {$Y$};
\end{tikzpicture}
\hspace{0.7cm}
\raisebox{-2pt}{$\xmapsto{\ \dg\ }$}
\hspace{0.7cm}
\begin{tikzpicture}[scale=1.0, baseline={(0,-0.1)}]
  \draw (0,0) circle (1.2);
  \draw[thick,->] (0,-1.2) -- (0,-0.15);
  \draw[thick,->] (0,0) -- (0,0.7);
  \draw[thick] (0,0.7) -- (0,1.2);
  \fill (0,0) circle (2pt); \node[right] at (0.07,0) {$\varphi^\dg$};
  \node[below] at (0,-1.25) {$Y$};
  \node[above] at (0,1.25) {$X$};
\end{tikzpicture}
\caption{The dagger of a 2-morphism given by a decorated disk: total orientation reversal, re-read as a state on the reflected circle. }
\label{fig:daggerdisk}
\end{figure}

\begin{proposition}[properties of the dagger]
\label{prop:daggerprops}
The dagger \eqref{eq:daggerdef} is anti-linear and involutive, $(\varphi^\dg)^\dg = \varphi$, and recovers the hermitian form through the cup $h$:
\begin{equation}
  \cZ(h)(\varphi^\dg \otimes \psi) = b_E(\varphi, \psi)
  \qquad \text{for all } \psi \in \operatorname{Hom}(X,Y) \, .
\end{equation}
\end{proposition}

\begin{proof}
\emph{Involutivity.} As in dimension one, this is immediate from the equivariance coherence $\overline{\rho} \circ (\rho \circ R) = \mathrm{id}$ of Definition~\ref{def:hermitiandefectTQFT}, together with $u_{Y,X} = R(u_{X,Y})^{-1}$---the statement that the reflection $z \mapsto \bar z$ of the standard circles (Remark~\ref{rem:canonicityu}) is an involution. 

\emph{Cup characterization.} With $b_E = \cZ(h_E) \circ (\rho_E^{-1} \otimes \mathrm{id})$ (Definition~\ref{def:RPproperty}) and $h = h_E \circ (u_{X,Y}^{-1} \sqcup \mathrm{id})$,
\begin{equation}
  \cZ(h)(\varphi^\dg \otimes \psi)
  = \cZ(h_E)\big(\cZ(u_{X,Y}^{-1})(\varphi^\dg) \otimes \psi\big)
  = \cZ(h_E)\big(\rho_E^{-1}(\overline{\varphi}) \otimes \psi\big)
  = b_E(\varphi, \psi) \, ,
\end{equation}
so the bent-cylinder pairing through $h$ recovers the hermitian form.
\end{proof}

\emph{Functoriality.}

\begin{lemma}
\label{lem:daggerfunctorial}
The maps \eqref{eq:daggerdef} assemble into a functor $\dg \colon \cT_\cZ \to \cT_\cZ^{2\op}$ which is the identity on objects and 1-morphisms:
\begin{enumerate}
  \item $(\Psi \circ \Phi)^\dg = \Phi^\dg \circ \Psi^\dg$ and $(1_X)^\dg = 1_X$;
  \item $(\widetilde{\Phi} \otimes \Phi)^\dg = \widetilde{\Phi}^\dg \otimes \Phi^\dg$ for horizontal composition.
\end{enumerate}
\end{lemma}

\begin{figure}[ht]
\centering
\begin{tikzpicture}[scale=0.78, >=stealth, baseline=(current bounding box.center)]
  \draw[thick] (0,0) circle (1.7);
  \draw[thick] (0,-0.7) circle (0.36);
  \draw[thick] (0,0.7) circle (0.36);
  \draw[thick] (0,-0.34) -- (0,0.34);
  \draw[thick] (0,-1.06) -- (0,-1.7);
  \draw[thick] (0,1.06) -- (0,1.7);
  \fill (0,-0.34) circle (1.2pt);
  \fill (0,0.34) circle (1.2pt);
  \fill (0,-1.06) circle (1.2pt);
  \fill (0,1.06) circle (1.2pt);
  \fill (0,-1.7) circle (1.2pt);
  \fill (0,1.7) circle (1.2pt);
  \node[right] at (0.1,0) {$Y$};
  \node[right] at (0.1,-1.4) {$X$};
  \node[right] at (0.1,1.4) {$Z$};
  \node at (-0.95,-0.7) {$E_{X,Y}$};
  \node at (-0.95,0.7) {$E_{Y,Z}$};
  \node[right] at (1.72,0) {$E_{X,Z}$};
\end{tikzpicture}
\qquad $\xrightarrow{\ z\,\mapsto\,\bar z\ }$ \qquad
\begin{tikzpicture}[scale=0.78, >=stealth, baseline=(current bounding box.center)]
  \draw[thick] (0,0) circle (1.7);
  \draw[thick] (0,0.7) circle (0.36);
  \draw[thick] (0,-0.7) circle (0.36);
  \draw[thick] (0,-0.34) -- (0,0.34);
  \draw[thick] (0,1.06) -- (0,1.7);
  \draw[thick] (0,-1.06) -- (0,-1.7);
  \fill (0,-0.34) circle (1.2pt);
  \fill (0,0.34) circle (1.2pt);
  \fill (0,-1.06) circle (1.2pt);
  \fill (0,1.06) circle (1.2pt);
  \fill (0,-1.7) circle (1.2pt);
  \fill (0,1.7) circle (1.2pt);
  \node[right] at (0.1,0) {$Y$};
  \node[right] at (0.1,1.4) {$X$};
  \node[right] at (0.1,-1.4) {$Z$};
  \node at (-0.95,0.7) {$E_{Y,X}$};
  \node at (-0.95,-0.7) {$E_{Z,Y}$};
  \node[right] at (1.72,0) {$E_{Z,X}$};
\end{tikzpicture}
\caption{The composition pants $P_{X,Y,Z} \colon E_{Y,Z} \sqcup E_{X,Y} \to E_{X,Z}$ (left), in which the $Y$-strands of the two incoming circles are joined while the $X$- and $Z$-strands run out to $E_{X,Z}$. The reflection $z \mapsto \bar z$ of Remark~\ref{rem:canonicityu} sends each boundary circle to its dual ($E_{X,Y}\mapsto E_{Y,X}$, $E_{Y,Z}\mapsto E_{Z,Y}$, $E_{X,Z}\mapsto E_{Z,X}$), exchanges the two incoming circles, and leaves the $Y$-strands joined, carrying the pants to $P_{Z,Y,X} \circ \sigma$ (right); see~\eqref{eq:pantsreflection}.}
\label{fig:pantsreflection}
\end{figure}
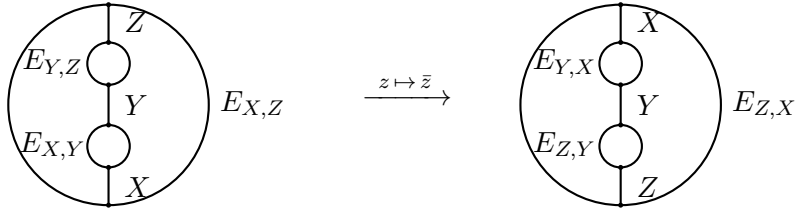

\begin{proof}
By Definition~\ref{def:daggeronTZ} the dagger is $\varphi^\dg = \cZ(u_{X,Y})\,\rho_{E}^{-1}(\overline{\varphi})$, so functoriality is naturality of $\rho$ read on the composition pants, once we record how that pants behaves under reflection.

\emph{The composition pants under reflection.} For composable $X, Y, Z \colon \alpha \to \beta$ let $P_{X,Y,Z} \colon E_{Y,Z} \sqcup E_{X,Y} \to E_{X,Z}$ be the vertical composition pants of Section~\ref{sec:defectTQFT} in which the $Y$-strands of the two incoming circles are joined while the $X$- and $Z$-strands run out to $E_{X,Z}$ (Figure~\ref{fig:pantsreflection}, left). Under the identifications $u$ of Lemma~\ref{lem:EXYdual},
\begin{equation}
\label{eq:pantsreflection}
  u_{X,Z} \circ R(P_{X,Y,Z}) \circ \big(u_{Y,Z}^{-1} \sqcup u_{X,Y}^{-1}\big) = P_{Z,Y,X} \circ \sigma \, ,
\end{equation}
with $\sigma$ the transposition of the two incoming circles. Indeed, the reflection $z \mapsto \bar z$ defining the $u$'s (Remark~\ref{rem:canonicityu}) sends each boundary circle to its dual---$E_{X,Y}\mapsto E_{Y,X}$, $E_{Y,Z}\mapsto E_{Z,Y}$ and $E_{X,Z}\mapsto E_{Z,X}$---keeps the two $Y$-strands joined, and exchanges the two incoming circles, carrying the trinion to $P_{Z,Y,X}\circ\sigma$ by a decoration-preserving diffeomorphism of the standard models (Figure~\ref{fig:pantsreflection}, right). The analogous identity, without the transposition $\sigma$, holds for the horizontal composition pants: its two incoming circles sit side by side and are not exchanged by the reflection.

\emph{Naturality.} For $\Psi \in \operatorname{Hom}(Y,Z)$ and $\Phi \in \operatorname{Hom}(X,Y)$, naturality and monoidality of $\rho$ give
\begin{equation}
  \rho_{E_{X,Z}}^{-1}\big(\overline{\cZ(P_{X,Y,Z})(\Psi \otimes \Phi)}\big)
  = \cZ\big(R(P_{X,Y,Z})\big)\big(\rho_{E_{Y,Z}}^{-1}(\overline{\Psi}) \otimes \rho_{E_{X,Y}}^{-1}(\overline{\Phi})\big) \, .
\end{equation}
Since $\rho_{E_{Y,Z}}^{-1}(\overline{\Psi}) = \cZ(u_{Y,Z}^{-1})(\Psi^\dg)$ and $\rho_{E_{X,Y}}^{-1}(\overline{\Phi}) = \cZ(u_{X,Y}^{-1})(\Phi^\dg)$, applying $\cZ(u_{X,Z})$ and \eqref{eq:pantsreflection} gives
\begin{equation}
  (\Psi \circ \Phi)^\dg = \cZ\big(u_{X,Z} \circ R(P_{X,Y,Z}) \circ (u_{Y,Z}^{-1} \sqcup u_{X,Y}^{-1})\big)(\Psi^\dg \otimes \Phi^\dg) = \cZ(P_{Z,Y,X})(\Phi^\dg \otimes \Psi^\dg) = \Phi^\dg \circ \Psi^\dg \, .
\end{equation}
The unit law $(1_X)^\dg = 1_X$ is then formal: $1_X^\dg \circ \Phi^\dg = (\Phi \circ 1_X)^\dg = \Phi^\dg$ for every $\Phi$, and $\dg$ is a bijection by Proposition~\ref{prop:daggerprops}, so $1_X^\dg$ is a two-sided unit at $X$ and hence equals $1_X$. Part~(2) is identical, with the horizontal composition pants in place of $P_{X,Y,Z}$.
\end{proof}

This completes Part~1: $\dg$ is an anti-linear dagger functor, the identity on objects and 1-morphisms, with $\dg^2 = \mathrm{id}$, and the form $(\varphi, \psi) \mapsto \cZ(h)(\varphi^\dg \otimes \psi)$ on $\operatorname{Hom}(X,Y)$ \emph{is} the hermitian form $b_{E_{X,Y}}$. Whether these forms are positive definite---reflection positivity---is taken up in Section~\ref{sec:positivity}.

\paragraph{Part 2: compatibility with duality.}

By Section~\ref{sec:defectTQFT} the bicategory $\cT_\cZ$ has adjoints (the rainbow $\mathrm{ev}_X, \mathrm{coev}_X$), so the left-mate dual functor $(-)^L$ of Definition~\ref{def:strictO2} is defined. Since $\dg$ is a functor (Part~1), exhibiting $\cT_\cZ$ as an $\Or(2)$-dagger bicategory reduces to the second datum of Definition~\ref{def:strictO2}: that $(-)^L$ commutes with $\dg$, with unitary coherences. The abstract consequences of Section~\ref{sec:strictO2} then follow automatically; below we supply only the geometric identifications and collect the consequences at the end.

\emph{Duality data.}

\begin{lemma}
\label{lem:unitaryduality}
$(\mathrm{ev}_X)^\dg = \widetilde{\mathrm{coev}}_X$ and $(\mathrm{coev}_X)^\dg = \widetilde{\mathrm{ev}}_X$.
\end{lemma}

\begin{proof}
Since $\dg$ is a functor (Lemma~\ref{lem:daggerfunctorial}), applying $\dg$ to the two Zorro identities for $(\mathrm{ev}_X, \mathrm{coev}_X)$ shows, exactly as in Lemma~\ref{lem:abstractunitaryduality}, that $(\mathrm{coev}_X)^\dg$ and $(\mathrm{ev}_X)^\dg$ satisfy the Zorro identities for an adjunction $X \dashv X^L$; that is, they are right adjunction data for $X$. It remains to identify them on the nose with the geometric tilde data of~\cite{Car16}. All four 2-morphisms are $\cZ$ of rainbow disks. By Definition~\ref{def:daggeronTZ} and naturality of $\rho$, the dagger of $\cZ$ of a disk is $\cZ$ of its total orientation reversal; thus $(\mathrm{ev}_X)^\dg$ is $\cZ$ of the evaluation rainbow disk with all strata orientations reversed, read on the reflected boundary circle. Reversing the orientations of the nested rainbow arcs and reflecting the disk, which reverses the order of the marked points, is precisely the prescription ``reverse all orientations and orders'' that defines $\widetilde{\mathrm{coev}}_X$ from $\mathrm{ev}_X$ in~\cite[eqs.~(2.18),(2.19)]{Car16}; under the boundary identification $u$ of Lemma~\ref{lem:EXYdual} the two decorated disks coincide, whence $(\mathrm{ev}_X)^\dg = \widetilde{\mathrm{coev}}_X$. The second identity is the same computation with the roles of $\mathrm{ev}$ and $\mathrm{coev}$ interchanged.
\end{proof}

This is the defect-bicategory analogue of the unitary dual structure of~\cite{Bar25}: the dagger exchanges left and right adjunction data. By Lemmas~\ref{lem:daggerfunctorial} and~\ref{lem:unitaryduality} it therefore carries the left mate $\varphi^L$ (via $\mathrm{ev}, \mathrm{coev}$) to the right mate of $\varphi^\dg$ (via the tilde data), so the compatibility $(\varphi^\dg)^L = (\varphi^L)^\dg$ amounts to the equality of left and right mates---that is, to pivotality.

\emph{Pivotality.} For $\varphi \in \operatorname{Hom}(X,Y)$ write
\begin{equation}
  \varphi^L = (\mathrm{ev}_Y \otimes \mathrm{id}) \circ (\mathrm{id} \otimes \varphi \otimes \mathrm{id}) \circ (\mathrm{id} \otimes \mathrm{coev}_X), \quad
  \varphi^R = (\mathrm{id} \otimes \widetilde{\mathrm{ev}}_Y) \circ (\mathrm{id} \otimes \varphi \otimes \mathrm{id}) \circ (\widetilde{\mathrm{coev}}_X \otimes \mathrm{id})
\end{equation}
for the left and right mates. Their agreement is the \emph{pivotality} of $\cT_\cZ$, a theorem of~\cite{DKR11}, with a transparent geometric origin: the mates are computed by transporting the defect lines of $E_{X,Y}$ through a half-turn---the nested arcs of the rainbow disks $\mathrm{ev}, \mathrm{coev}$ carry them by $+\pi$ and the tilde data $\widetilde{\mathrm{ev}}, \widetilde{\mathrm{coev}}$ by $-\pi$ (Figure~\ref{fig:annuli})---and the two differ by a full rotation; the full-twist annulus is the image of the identity cylinder under a Dehn twist and hence equal to it in $\Bord_2^{\textrm{def}}(\bD)$, so that $\varphi^L = \varphi^R$. (This uses that morphisms are bordisms taken up to diffeomorphism relative to the boundary, as in~\cite{DKR11}.) In the same picture a reflection conjugates the $+\pi$ rotation to the $-\pi$ one, exhibiting pivotality and the compatibility of $\dg$ with $(-)^L$ as the planar shadow of the $\Or(2)$-fixed-point description.

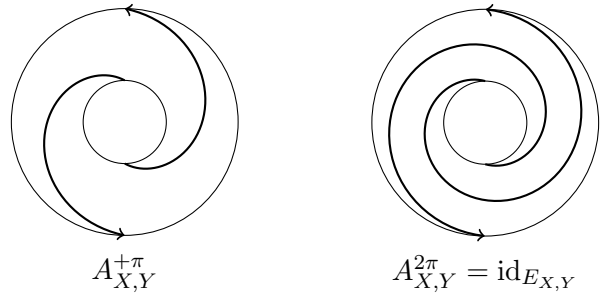
\begin{figure}[ht]
\centering
\begin{tikzpicture}[scale=1.0]
  \draw (0,0) circle (1.5);
  \draw (0,0) circle (0.55);
  \draw[thick,->] plot[domain=0:1, samples=50] ({(0.55+0.95*\x)*cos(90+180*\x)}, {(0.55+0.95*\x)*sin(90+180*\x)});
  \draw[thick,->] plot[domain=0:1, samples=50] ({(0.55+0.95*\x)*cos(270+180*\x)}, {(0.55+0.95*\x)*sin(270+180*\x)});
  \node at (0,-2) {$A^{+\pi}_{X,Y}$};
\end{tikzpicture}
\hspace{1.5cm}
\begin{tikzpicture}[scale=1.0]
  \draw (0,0) circle (1.5);
  \draw (0,0) circle (0.55);
  \draw[thick,->] plot[domain=0:1, samples=80] ({(0.55+0.95*\x)*cos(90+360*\x)}, {(0.55+0.95*\x)*sin(90+360*\x)});
  \draw[thick,->] plot[domain=0:1, samples=80] ({(0.55+0.95*\x)*cos(270+360*\x)}, {(0.55+0.95*\x)*sin(270+360*\x)});
  \node at (0,-2) {$A^{2\pi}_{X,Y} = \mathrm{id}_{E_{X,Y}}$};
\end{tikzpicture}
\caption{The half-twist annulus computing the mates (left), transporting the defect lines of $E_{X,Y}$ through $\pm\pi$, and the full-twist annulus (right), the image of the identity cylinder under a Dehn twist and hence equal to it. Two defect lines are drawn; in general all lines of $E_{X,Y}$ are transported.}
\label{fig:annuli}
\end{figure}

\begin{proposition}[pivotality and the compatibility of $\dg$ and $(-)^L$]
\label{prop:pivotal}
The bicategory $\cT_\cZ$ is strictly pivotal---$\varphi^L = \varphi^R$ for every 2-morphism $\varphi$, and the canonical pivotal structure $\delta$ of Proposition~\ref{prop:abstractpivotal} is the identity---and $(-)^L$ commutes with $\dg$: $(\varphi^\dg)^L = (\varphi^L)^\dg$ for all $\varphi$, the second datum of Definition~\ref{def:strictO2}.
\end{proposition}

\begin{proof}
Pivotality $\varphi^L = \varphi^R$ is~\cite{DKR11}, whose geometric origin was sketched above; it is strict because $X^{LL} = X$ on the nose---orientation reversal of the defect lines being an involution---so $\delta = \mathrm{id}$. The compatibility is then formal: by Lemmas~\ref{lem:daggerfunctorial} and~\ref{lem:unitaryduality} the dagger carries the left mate of $\varphi$ to the right mate of $\varphi^\dg$, that is $(\varphi^L)^\dg = (\varphi^\dg)^R$, and pivotality applied to $\varphi^\dg$ gives $(\varphi^\dg)^L = (\varphi^\dg)^R = (\varphi^L)^\dg$.
\end{proof}

This completes the proof of Theorem~\ref{thm:extraction}: Part~1 provides the dagger $\dg$ and Part~2 the dual functor $(-)^L$, compatible by Proposition~\ref{prop:pivotal}. The comparison $2$-morphisms $\nu, \iota$ of Definition~\ref{def:strictO2} are unitary, being values of $\cZ$ on the isotopy cylinders relating the nested and concatenated rainbow configurations of Section~\ref{sec:defectTQFT}, whose daggers are the reverse isotopy cylinders; thus $\cT_\cZ$ is an $\Or(2)$-dagger bicategory. The abstract consequences of Section~\ref{sec:strictO2} now specialize to $\cT_\cZ$: the daggered, two-sided adjunction data of Lemma~\ref{lem:abstractunitaryduality} (identified with the tilde data of~\cite{Car16} in Lemma~\ref{lem:unitaryduality}), the canonical unitary pivotal structure of Proposition~\ref{prop:abstractpivotal} (trivial, by Proposition~\ref{prop:pivotal}), the conjugation $(-)^* = \dg^b$ of Lemma~\ref{lem:conjugation}, and hence the bi-involutive structure of~\cite{HP17} on each endomorphism category $\cT_\cZ(\alpha,\alpha)$.

\section{Reflection positivity}
\label{sec:positivity}

The construction of Section~\ref{sec:extractionsection} used only the reflection structure $(\cZ, \rho)$ of Definition~\ref{def:hermitiandefectTQFT}. We now add positivity.

Suppose $(\cZ, \rho)$ is \emph{reflection positive} (Definition~\ref{def:RPproperty}): the form $b_E$ is positive definite for every object $E$. Equivalently, by~\cite[Thm.~2.3.41]{Ste23}, $\cZ$ is a symmetric monoidal dagger functor to $\Hilb$ (Definition~\ref{def:RPdefectTQFT}). Then every $2$-morphism space $\operatorname{Hom}(X,Y) = \cZ(E_{X,Y})$ is a Hilbert space, the pairing $(\varphi, \psi) \mapsto \cZ(h)(\varphi^\dg \otimes \psi) = b_{E_{X,Y}}(\varphi, \psi)$ of Theorem~\ref{thm:extraction} is its inner product, with respect to which the dagger $\dg \colon \operatorname{Hom}(X,Y) \to \operatorname{Hom}(Y,X)$ is anti-unitary. As noted in Remark~\ref{rem:positivityforfree}, in the $\Hilb$ formulation this positivity is automatic, wired into the target. 

We now explain the connection to a concept introduced in~\cite{CHFHS24}: 
\begin{definition} 
	Let $\cB$ be an $\Or(2)$-dagger bicategory. A \emph{spherical weight} is a collection of linear functionals $\psi_\alpha \colon \operatorname{End}(1_\alpha) \to \bC$, one for each object $\alpha$, satisfying $\psi_\alpha(f^\dagger f) \geq 0$ for all $f \in \operatorname{End}(1_\alpha)$, with equality if and only if $f = 0$, and such that, for every $1$-morphism $X \colon \alpha \to \beta$ and every $2$-morphism $f \colon X \Rightarrow X$, the right and left traces of Definition~\ref{def:traces} satisfy
	\begin{align}
		\psi_\alpha (\operatorname{tr}^R(f)) = \psi_\beta (\operatorname{tr}^L(f)) \, .
	\end{align}
\end{definition}

\begin{figure}[ht]
\centering
\includegraphics[width=0.7\textwidth]{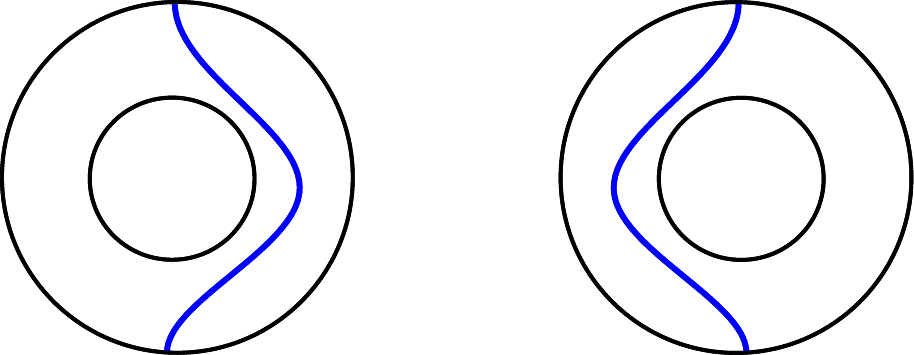}
\caption{The bordisms representing the right trace $\operatorname{tr}^R$ (left) and the left trace $\operatorname{tr}^L$ (right), to be read as morphisms from the outer circle to the inner one. Filling in the inner disk corresponds to applying the spherical weight; the two bordisms then agree by an isotopy, showing that the right and left traces coincide once the spherical weight is applied.}
\label{fig:spherical}
\end{figure}

This is the essential extra datum of a 3-Hilbert space beyond the $\Or(2)$-dagger structure: together with the Hilbert space structure on the $2$-morphism spaces coming from positivity, it makes all hom-categories 2-Hilbert spaces~\cite{CHFHS24}. The spherical weight has a natural interpretation in terms of the defect TQFT: it is the partition function $\cZ$ assigns to the sphere with the point defect $f$ inserted. Since we only have indirect access to the point defects through the $D_0$-completion, we now give the precise connection in terms of the data more directly available to us.
\begin{proposition}
Let $\cT_\cZ$ be the $\Or(2)$-dagger bicategory associated to a reflection positive defect TQFT. The linear maps
\begin{align}
	\psi_\alpha\colon \operatorname{End}(1_\alpha) = \cZ(S^1_\alpha) & \to \bC \\ 
	\varphi &\mapsto b_{S^1_\alpha}(\mathrm{id}, \varphi)
\end{align}
define a spherical weight.
\end{proposition}
\begin{proof}
	Write $\varphi^\dagger\varphi$ for the vertical composite in $\operatorname{End}(1_\alpha) = \cZ(S^1_\alpha)$. Inserting it on the capped circle yields $\cZ$ of the sphere obtained by gluing the disk carrying $\varphi$ to its dagger along $S^1_\alpha$; reading this same sphere as the cup pairing of $\varphi$ with itself gives $b_{S^1_\alpha}(\mathrm{id}, \varphi^\dagger\varphi) = b_{S^1_\alpha}(\varphi, \varphi)$. Positivity and faithfulness of $\psi_\alpha$ then follow from the positive-definiteness of $b$. Finally, $\psi_\alpha(\operatorname{tr}^R(f)) = \psi_\beta(\operatorname{tr}^L(f))$ follows from a geometric argument sketched in Figure~\ref{fig:spherical} and explained in its caption.
\end{proof}

Beyond the spherical weight, a full 3-Hilbert space requires some completeness and finiteness conditions, such as the existence of direct sums. These are expected to hold for topological defects in physical theories; this, for example, underlies the widely held belief that the line defects of a 2-dimensional TQFT form a fusion category~\cite{BT17}. Granting them, the defect bicategory $\cT_\cZ$ of a reflection positive defect TQFT is a 3-Hilbert space.

\end{document}